\documentclass[a4paper,noarxiv,onecolumn,11pt]{quantumarticle}
\pdfoutput=1
\usepackage[numbers,sort&compress]{natbib}
\usepackage[utf8]{inputenc}
\usepackage[english]{babel}
\usepackage[T1]{fontenc}
\usepackage{amsmath,amssymb,bm,amsthm}
\usepackage{hyperref}

\usepackage{tikz}
\usepackage{lipsum}
\usepackage{tabularx}
\usepackage{subfig}
\usepackage{caption}
\usepackage{graphicx}
\usepackage{floatrow}

\usepackage[normalem]{ulem}

\usepackage[ruled, lined, linesnumbered, commentsnumbered, longend]{algorithm2e}
\usepackage{xcolor}

\SetCommentSty{mycommfont}

\newtheorem{theorem}{Theorem}
\newtheorem{lemma}{Lemma}

\newtheorem{definition}{Definition}

\usepackage{mathtools}
\DeclarePairedDelimiter\bra{\langle}{\rvert}
\DeclarePairedDelimiter\ket{\lvert}{\rangle}
\DeclarePairedDelimiterX\braket[2]{\langle}{\rangle}{#1\,\delimsize\vert\,\mathopen{}#2}

\DeclarePairedDelimiter\ceil{\lceil}{\rceil}

\usepackage{color}

\DeclareMathOperator*{\argmax}{arg\,max}

\begin{document}

\title{Coreset selection can accelerate quantum machine learning models with provable generalization}

\author{Yiming Huang}
\affiliation{Center on Frontiers of Computing Studies, Peking University, Beijing 100871, China}
\affiliation{School of Computer Science, Peking University, Beijing 100871, China}

\author{Huiyuan Wang}
\affiliation{Peterhouse, Univeristy of Cambridge, CB2 1RD Cambridge, Cambridgeshire, U.K.}

\author{Yuxuan Du}
\affiliation{JD Explore Academy, Beijing 101111, China}

\author{Xiao Yuan}
\affiliation{Center on Frontiers of Computing Studies, Peking University, Beijing 100871, China}
\affiliation{School of Computer Science, Peking University, Beijing 100871, China}


\maketitle

\begin{abstract}
Quantum neural networks (QNNs) and quantum kernels stand as prominent figures in the realm of quantum machine learning, poised to leverage the nascent capabilities of near-term quantum computers to surmount classical machine learning challenges. Nonetheless, the training efficiency challenge poses a limitation on both QNNs and quantum kernels, curbing their efficacy when applied to extensive datasets. To confront this concern, we present a unified approach: coreset selection, aimed at expediting the training of QNNs and quantum kernels by distilling a judicious subset from the original training dataset. Furthermore, we analyze the generalization error bounds of QNNs and quantum kernels when trained on such coresets, unveiling the comparable performance with those training on the complete original dataset. Through systematic numerical simulations, we illuminate the potential of coreset selection in expediting tasks encompassing synthetic data classification, identification of quantum correlations, and quantum compiling. Our work offers a useful way to improve diverse quantum machine learning models with a theoretical guarantee while reducing the training cost.
\end{abstract}

\section{Introduction}

Quantum neural networks (QNNs) \cite{schuld2014quest,farhi2018classification,havlivcek2019supervised,cong2019quantum} and quantum kernels \cite{schuld2019quantum,Huang2021PowerOD} have emerged as pivotal models in the burgeoning field of quantum machine learning (QML) 		\cite{Biamonte2017QuantumML,Benedetti2019ParameterizedQC,cerezo2021variational}, poised to unlock the power of near-term quantum computers to address challenges that elude classical machine learning paradigms \cite{preskill2018quantum,bharti2021noisy}. The allure of these models is rooted in a fusion of theoretical advances and practical adaptability. That is, theoretical evidence showcases their superiority over classical counterparts in diverse scenarios, spanning synthetic datasets, discrete logarithmic problems, and quantum information processing tasks \cite{Huang2021PowerOD,Liu2021ARA,abbas2021power,jager2023universal,du2022demystify,wu2023quantum}, as measured by sample complexity and runtime considerations. Complementing their theoretical strength, their implementation displays flexibility, adeptly accommodating constraints posed by contemporary quantum hardware, including qubit connectivity and limited circuit depth. This convergence of theoretical promise and practical flexibility has spurred a wave of empirical investigations, substantiating the viability and potential benefits of QNNs and quantum kernels across real-world applications such as computer vision \cite{peters2021machine,huang2021experimental,pan2023experimental} and quantum physics \cite{ren2022experimental,herrmann2022realizing}.

Despite their promising potential, QNNs and quantum kernels confront a pertinent challenge concerning the training efficiency, resulting in a constrained practical applicability towards large-scale datasets \cite{schuld2019evaluating}. This limitation is particularly evident due to the absence of fundamental training mechanisms like back-propagation and batch gradient descent in the majority of QNNs, imperative for the swift training of deep neural networks \cite{Goodfellow2015DeepL}. Similarly, the training process of quantum kernels necessitates the collection of a kernel matrix of size $O(N^2)$, with $N$ being the number of training examples and each entry demanding independent evaluation via a specific quantum circuit. Consequently, the capacities of both QNNs and quantum kernels to effectively navigate vast training datasets, characterized by millions of data points, are compromised.

In response to the above challenge, several research lines have emerged to enhance the training efficiency of QNNs. The first line embarks on improving the optimizer or the initialization methods, seeking to expedite convergence towards the minimal empirical risk through a reduction in measurements and iterations \cite{verdon2019learning,stokes2020quantum,gacon2021simultaneous,van2021measurement,volkoff2021large}. Nonetheless, the non-convex nature of the loss landscape cautions against potential entrapment within saddle points during the optimization process \cite{sweke2020stochastic,du2021learnability,bittel2021training}. The second avenue delves into feature dimension reduction techniques, enabling more streamlined utilization of quantum resources for each data point \cite{wilson2018quantum,hur2022quantum}. However, this approach may not necessarily alleviate overall runtime complexity, potentially even exacerbating it. The third research pathway navigates the realm of QNN expressivity, imposing judicious constraints to facilitate the integration of back-propagation through meticulously engineered ansatzes \cite{bowles2023backpropagation,abbas2023quantum}. Nevertheless, the extent to which these constrained ansatzes might impact QNN performance on unseen data remains a dynamic yet unresolved facet.    
 
The endeavor to enhance the training efficiency of quantum kernels has received less attention in contrast to QNNs. This divergence in attention stems from the shared observation that both classical and quantum kernels demand $O(N^2)$ runtime for the collection of the kernel matrix. Existing literature focused on enhancing the training efficiency of quantum kernels has predominantly centered on the development of advanced encoding methods \cite{glick2021covariant,hubregtsen2022training}. These methods aim to mitigate the usage of quantum resources and attenuate the manifestation of barren plateaus \cite{thanasilp2022exponential}. However, despite the cultivation of various research trajectories directed at enhancing QML models, it is noteworthy that these approaches often retain model-specific attributes, potentially harboring unforeseen side effects. This realization begets a pivotal inquiry: Can a unified approach be devised that systematically enhances the training efficiency of both QNNs and quantum kernels while safeguarding a steadfast theoretical guarantee?

In this study, we provide an affirmation of the above question by introducing the coreset selection techniques into QML. Conceptually, coreset is an effective preprocessing approach to distill a weighted subset from the large training dataset, which guarantees that models fitting the coreset also provide a good fit for the original data. Considering that the training efficiency of both QNNs and quantum kernels hinges on the number of training examples, coreset provides a unified way to enhance their training efficiency.  Moreover, the theoretical foundation of coreset is established on the recent generalization error analysis of QNNs and quantum kernels, indicating that few training examples are sufficient to attain a good test performance. In this regard, we provide a rigorous analysis of the generalization ability of QNNs and quantum kernels training on the coreset.  The achieved bounds exhibit the comparable generalization performance of QNN and quantum kernel learning when they are optimized under the original dataset and the coreset. Numerical simulations on synthetic data, identification of non-classical correlations in quantum states, and quantum circuit compilation confirm the effectiveness of our proposal.

\section{Related works}
The prior literature related to our work can be divided into two classes: the algorithms for accelerating the optimization of QML models and the generalization error analysis of QML models. In the following, we separately explain how our work relates to and differs from the previous studies.

\noindent\textbf{Acceleration algorithms for QML models}. As aforementioned, various algorithms have been introduced to expedite the optimization of QNNs rather than quantum kernels. These algorithms can be classified into three distinct categories, each addressing different facets of optimization improvement: data feature engineering, QNN architecture design, and optimizer enhancement.

In the realm of data feature engineering, the core concept revolves around implementing dimensionality reduction techniques such as principal component analysis and feature selection during the pre-processing stage \cite{bishop2007}. This strategy effectively reduces the quantum resources required for processing data points compared to their unprocessed counterparts. Within the domain of architecture design, there exists a dual focus on encoding strategy design \cite{perez2020data,lloyd2020quantum} and ansatz design \cite{bilkis2021semi,zhang2022differentiable,du2022quantum,sauvage2022building,larocca2022group}. The underlying principle emphasizes the utilization of a minimal number of qubits, trainable parameters, and shallower circuit depths, all geared towards achieving competent performance in learning tasks. In parallel, efforts to enhance optimizers have also garnered attention. The deployment of higher-order gradient descent or machine-learning assisted optimizers \cite{verdon2019learning,stokes2020quantum,gacon2021simultaneous,van2021measurement,volkoff2021large} and distributed learning schemes \cite{du2022distributed} has been advocated as a means to accelerate convergence rates and wall-clock time, thereby further augmenting the efficiency of QNN optimization.  

Our work is complementary to the above approaches since the reduced size of data is independent of data feature engineering, QNN architecture design, and optimizer enhancement. In other words, QNNs trained on coreset can be further accelerated by the above approaches. Moreover, different from prior literature that concentrates on the acceleration of QNNs, coreset selection can be directly employed to accelerate the collection of quantum kernels.   

\noindent\textbf{Generalization of QML models}. Several studies have undertaken the task of quantifying the generalization error of QNNs through the lens of the foundational learning-theoretic technique known as uniform convergence \cite{banchi2021generalization,caro2021encoding,gyurik2021structural,du2022efficient,caro2022generalization}. In general, the resulting bound for generalization adheres to the format of $O(\sqrt{p/N})$, where $p$ denotes the count of trainable parameters and $N$ signifies the number of training instances. For quantum kernels, the generalization error upper bound scales with $O(\sqrt{C_1/N})$, where $C_1$ depends on the labels of data and the value of quantum kernel \cite{Huang2021PowerOD}. When system noise is considered, the generalization error upper bound degrades to  $O(\sqrt{C_1/N}+N/(C_2\sqrt{m}))$, where $m$ is the shot number and $C_2$ depends on the kernel value, shot number, and the noise level \cite{wang2021towards}. 

Our work leverages the above analysis to quantify how the coreset selection effects the learning performance of QNNs and quantum kernels, respectively.

\smallskip

We note that although Ref.~\cite{xue2023near} also discusses quantum coreset. Their primary emphasis lies in employing fault-tolerant quantum computers to speed up the construction of coresets, a subject that falls beyond the scope of our work. In addition, Ref.~\cite{qu2022performance} and Ref.~\cite{otgonbaatar2021assembly} illustrated the potential of coreset in specific applications, i.e., clustering and image classification, without a unified view and generalization error analysis.

\section{Preliminaries}
In this section, we first introduce the concept of machine learning and corerset. Then, we formally introduce the mechanism of quantum neural networks (QNNs) and quantum kernels under the supervised learning paradigm.

\subsection{Foundations of machine learning}\label{sec:pred-ML}
Let ${\cal{S}}_t=\{{{\bm{x}}}_i,y_i\}_{i=1}^{N_t}$ be the dataset in which each paired training example $({{\bm{x}}}_i,y_i)$ is independent and identically sampled over $\cal{Z}=\cal{X}\times\cal{Y}$ with probability distribution $p_{\cal{Z}}$, where $\cal{X}$ is feature space and $\cal{Y}$ the label space. The goal of supervised learning algorithms is to find a hypothesis $f:\cal{X}\rightarrow \cal{Y}$ with trainable parameters ${\bm{w}}$ such that the true risk $R$ on the distribution $p_{\cal{Z}}$ is minimized with
\begin{equation}
R =   \mathbb{E}_{{({{\bm{x}}},y)}\sim p_{\cal{Z}}}[l(f_{\bm{w}}(\bm{x}),y)],
\end{equation}
where $l(\cdot,\cdot)$ is the loss function used to measure the degree of fit between the output of the hypothesis and its corresponding ground truth.  As the distribution $p_{\cal{Z}}$ is unknown and, even given, accessing all the data over $\cal{Z}$ would be impractical, in practice, the optimal hypothesis is estimated by optimizing $\bm{w}$ to minimize  the empirical risk $R_e$ over the training dataset, i.e.,
\begin{equation}
R_e = \frac{1}{N_t}\sum_{({{\bm{x}}}_i,y_i)\in \cal{S}}l(f_{\bm{w}}({{\bm{x}}}_i) ,y_i).
\end{equation}

\subsection{Coreset in machine learning}

As the learning algorithm evolves, not only are models becoming increasingly complex, but training datasets are also growing larger. With growing volumes of data, the challenges on how to organize and analyze the massive data force us to devise effective approaches to condense the enormous data set. The concept of coreset selection, as a paradigm for extracting a data subset that is comprised of informative samples such that the generalization performance of the model trained over this reduced set is closed to those trained on the entire data set \cite{agarwal2005geometric}, is a promising solution to address the issue.
\begin{definition}{\bf{(Coreset)}}\label{def:coreset}
Let $P$ be a set of points in space $\mathcal{V}$, and $f$ be a monotone measure function, we call a subset $Q\subseteq P$ as an $\epsilon$-coreset of $P$, if 
\begin{equation}
| f( Q) - f(P) | \leq \epsilon \cdot f(P).
\end{equation}
\end{definition}

Various coreset selection approaches are exploited to assist in dealing with computationally intractable problems in different learning tasks and data types \cite{bachem2017practical,feldman2020introduction}. Throughout the whole work, we consider a geometry-based method to construct coreset  \cite{farahani2009facility}. As shown in Fig.~\ref{pic:kcp}, the goal of coreset construction is to find $k$ data points as the centers $\mathcal{C}$ such that the maximum distance between any point $s\in \mathcal{S}$ to its nearest center is minimized, i.e., selecting $\mathcal{C}$ such that the radius $\delta_c$ of which is minimized. The optimization of finding the coreset can be formulated as
\begin{equation}\label{eq:k-center}
\mathcal{S}_c = \arg\min_{\mathcal{C} \subseteq \mathcal{S}, |\mathcal{C}|=k} \max_{{{\bm{x}}}_j\in \mathcal{S}}D({{\bm{x}}}_j,{{\bm{x}}}_c),
\end{equation}
where $D({{\bm{x}}}_i,{{\bm{x}}}_c)=\min_{{{\bm{x}}}_c \in \mathcal{C}}d({{\bm{x}}}_i,{{\bm{x}}}_c)$ denotes the distance between the point $i$ to its closest center. Although it is an NP-Hard problem to find $\mathcal{S}_c$, there is a provable greedy algorithm can efficiently get a 2-approximate solution, i.e., if $\mathcal{C}^*$ is the optimal solution of Eq. (\ref{eq:k-center}), we can efficiently find a solution $\mathcal{C}$ such that $\delta_{\mathcal{C}^*} \leq \delta_{\mathcal{C}} \leq 2\cdot\delta_{\mathcal{C}^*}$.

\begin{figure}[H]
\centering
\includegraphics[width=8cm,height=5.5cm]{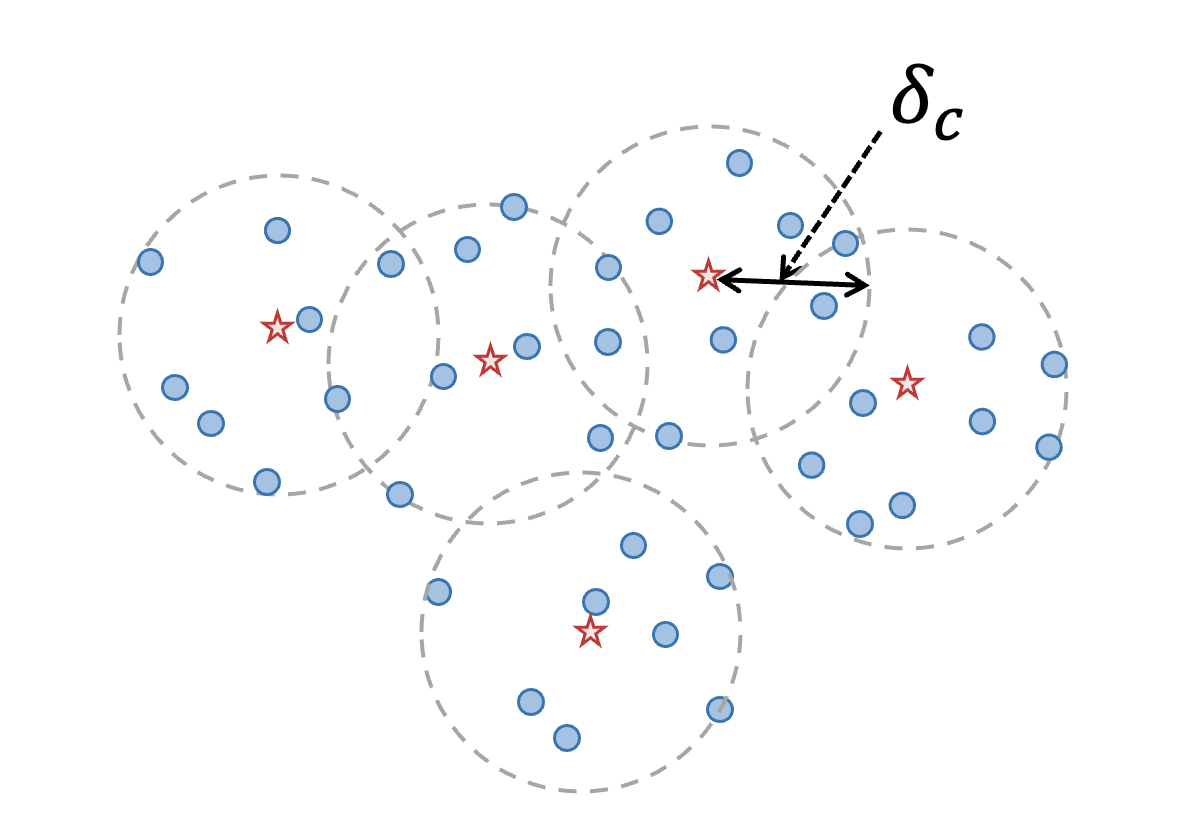}
\caption{Coreset construction as a $k$-center problem. The red stars are $k$ picked points with $\delta_c$ radius covering the entire set. In the case shown in the figure, there are 5 center data points $P_k\in \mathcal{S}$ with $k\in\{1,2,...,5\}$ such that the maximum distance from any point in $\mathcal{D}$ to its closest center is minimized.}
\label{pic:kcp}
\end{figure}

\subsection{Quantum neural networks}\label{sec:qnn}
 
Quantum neural network (QNN) refers to a class of neural networks that leverages the power of variational quantum circuits and classical optimizers to tackle learning problems. A QNN is mainly composed of three components: feature encoding, variational parametrized circuit, and measurement, as depicted in Fig.~\ref{pic:qnn}. Generally, feature encoding utilizes an encoding circuit $U_{\bm{x}}$ to map the classical input $\bm{x}$ into an $n$-qubit state $|\bm{x}\rangle$. The concrete approaches of feature encoding are diverse, as outlined in \cite{lloyd2020quantum,jerbi2023quantum}. A variational parametrized circuit $U_{\theta}$ evolves the feature state $|\bm{x}\rangle$ to a parametrized state $|\bm{x},\theta\rangle$, where the parameters $\theta$ are to be tuned for minimizing the training loss. Measurement extracts the processed information stored in parametrized state $|\bm{x},\theta\rangle$ into the classical register and may combine with a post-processing operation to form the output of QNN. 

\begin{figure}[H]
\centering
\includegraphics[width=12cm,height=4.5cm]{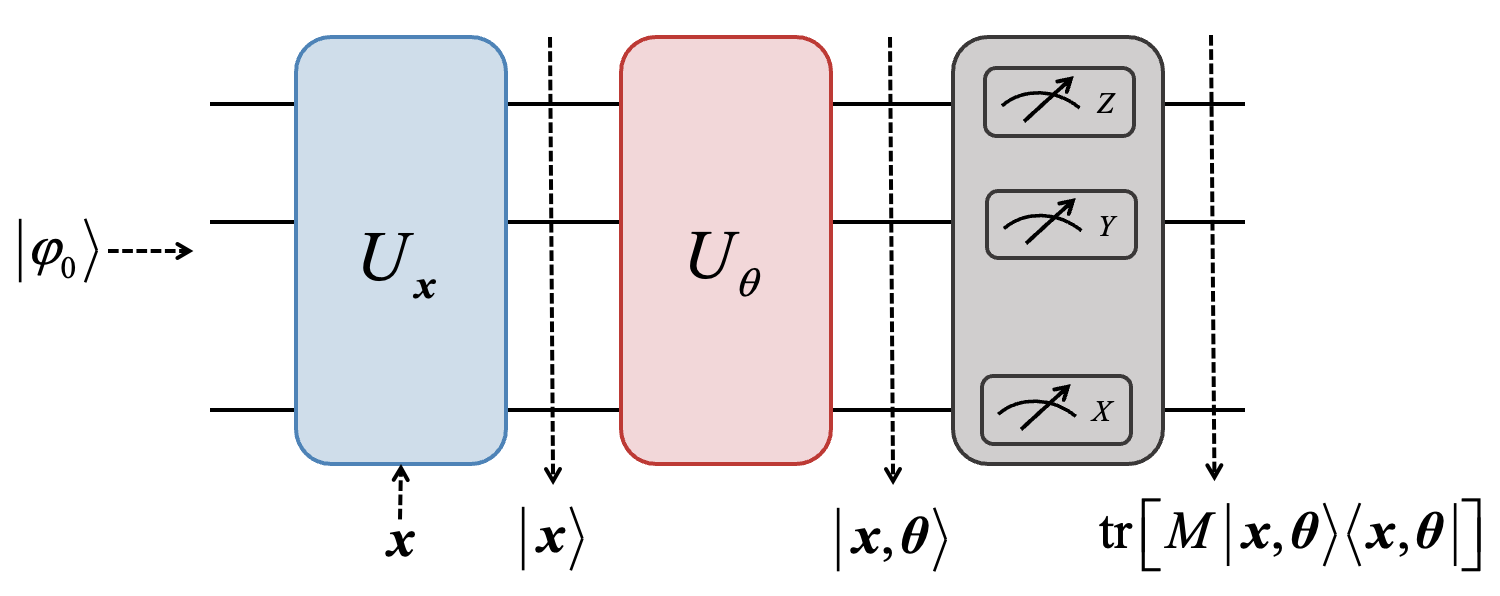}
\caption{A general feed-forward process of quantum neural networks. The input state $|\phi_0\rangle $ is firstly fed to feature encoding block to map the classical data $\bm{x}$ to $|\bm{x}\rangle$, and then use variational circuit $U_\theta$ to form a parametrized state $|\bm{x},\theta\rangle$ which is used for minimizing the loss function. In the end, the measurement output will be used to estimate the prediction residuals.}
\label{pic:qnn}
\end{figure}

In this work, we consider a QNN implemented by the data-reuploading strategy \cite{jerbi2023quantum,perez2020data,schuld2021effect}, alternating the feature encoding circuit $U_{\bm{x}}$ and the variational circuit $U_{\theta}$ to generate parametrized state $|\bm{x},\theta\rangle$, i.e., 
\begin{equation}
	|\bm{x},\theta\rangle =U_{\theta}U_{\bm{x}}... U_{\bm{x}} U_{\theta} U_{\bm{x}}\ket{\varphi}.
\end{equation}
Without loss of generality, the  feature encoding circuit and the variational circuit take the form as
\begin{equation}
	U_{\bm{x}}  = \bigotimes_{j=1}^n \exp{(-i\bm{x}_jH)}~\text{and}~ U_{\theta} = \bigotimes_{j=1}^n \exp(-i\theta_j H)V,
\end{equation} 
where $H \in \{\sigma_x,\sigma_y,\sigma_z\}$ is the Hermitian operator and $V$ refers to the fixed gates such as a sequence of CNOT gates. Once the variational state $|\bm{x},\theta\rangle$ is evolved, the measurement operator $M$ is applied to obtain the estimated expectation value. Given the training data set $\mathcal{S} = \{\bm{x}_i,y_i\}_{i=1}^{N_t}$, the empirical risk of QNN over $\mathcal{S}$ is 
\begin{equation}\label{eqn:sec:QNN-preliminary}
    R_e^{\text{QNN}}(\theta) = \frac{1}{|\mathcal{S}|}\sum_{\bm{x}_i,y_i\in \mathcal{S}}l(tr[M|\bm{x}_i,\theta\rangle\langle\bm{x}_i,\theta|],y_i).
\end{equation}
A possible choice of $l$ is the mean square error, i.e., $l(a, b)=(a-b)^2$. The optimization of QNN, i.e., the minimization of $ R_e^{\text{QNN}}$, can be completed by the gradient-descent optimizer based on the parameter shift rule  \cite{schuld2019evaluating}.  

\subsection{Quantum kernels}

The kernel methods have been extensively studied in classical machine learning and applied to various scenarios such as face recognition and the interpretation of the dynamics of deep neural networks  \cite{yang2001face,jacot2018neural}. Their quantum counterparts have also been extensively investigated from both experimental and theoretical perspectives \cite{havlivcek2019supervised,wang2021towards}. Formally, quantum kernels leverage quantum feature maps that encode the classical vector $\bm{x}$ into a higher Hilbert space to perform the kernel trick. One well known quantum kernel function $\kappa(\bm{x},\bm{x}')$ is defined as the overlap between the quantum states $|\bm{x}\rangle$ and $|\bm{x}'\rangle$ that encodes $\bm{x}$ and $\bm{x}'$ via an $n$-qubit feature encoding circuit $U_e$, i.e. $\kappa(\bm{x},\bm{x}')=|\langle\bm{x}'|\bm{x}\rangle|^2$ with $|\bm{x}\rangle = U_{e}(\bm{x})|+\rangle^{\otimes N}$ and $|+\rangle=H\ket{0}$, as shown in Fig.~\ref{pic:qkernel}.  In this work, we consider a generic  quantum feature map proposed in \cite{havlivcek2019supervised},  i.e.,
\begin{equation}
U_e(\bm{x}) =  \exp{(\sum_i\bm{x}_i\sigma_j^{Z}+\sum_{i,j}(\pi-\bm{x}_i)(\pi-\bm{x}_j)\sigma_i^{Z}\sigma_j^{Z})}.
\end{equation}
We note that other symmetric positive-definite functions are also valid candidates for implementing quantum kernels. Different forms of the quantum kernel correspond to different feature maps in Hilbert space, which could provide quantum merits for classification tasks if designed properly \cite{J_ger_2023}.

\begin{figure}[H]
\centering
\includegraphics[width=10cm,height=4cm]{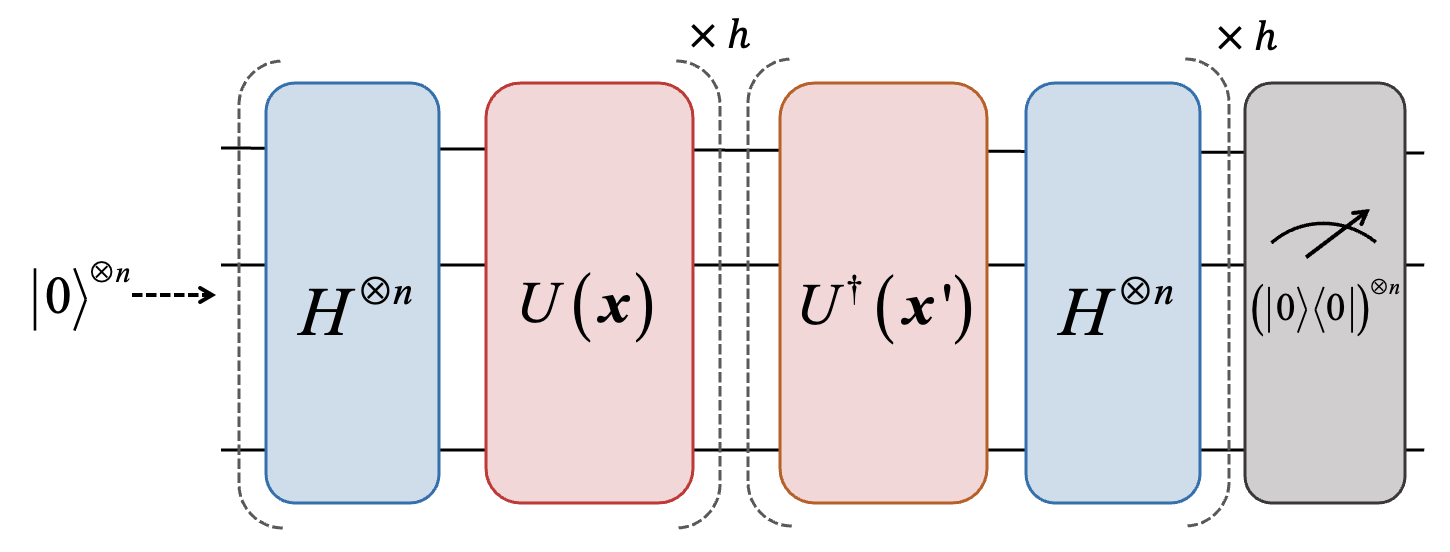}
\caption{The quantum circuit implementation of quantum kernel $\kappa{(\bm{x},\bm{x}')}$ considered in this work. The circuit separately encodes $\bm{x}$ and $\bm{x}'$ as the parameter of $h$ layers variational circuit $U$ and $U^{\dagger}$ into the states $|\bm{x}\rangle$ and $|\bm{x}'\rangle$ .}
\label{pic:qkernel}
\end{figure}


\section{Coreset selection for QML models}

This section first introduces the algorithmic implementation of coreset selection. Then, we elucidate how to use the constructed coreset to train QNNs and quantum kernels.

The Pseudo code of coreset selection for QML models is summarized in   Algorithm~\ref{alg:coreset}.  The main idea is regarding the coreset selection as a data preprocessing strategy, which aims to find $k$ balls, centered at points $\{\bm{x}_s\}_{s=1}^k$ in data space $\mathcal{P}$, to cover the whole data set with the radius $\delta_c$. Let  $\mathcal{S}=\{(\bm{x}_1,y_1),\dots, (\bm{x}_{N_t},y_{N_t})\}$ be an $n_c$-class data set and a hyper-parameter $k$ satisfy $k \gg n_c$. We apply the following procedure to each class to construct the coreset.\\

\IncMargin{2em}
\begin{algorithm}[H]
\SetKwData{Left}{left}\SetKwData{This}{this}\SetKwData{Up}{up}
\SetKwFunction{Union}{Union}\SetKwFunction{FindCompress}{FindCompress}
\SetKwInOut{Input}{input}\SetKwInOut{Output}{output}

\Input{Data set $\mathcal{S}=\{(\bm{x}_1,y_1),\dots, (\bm{x}_{N_t},y_{N_t})\}$, covering number $k$.}
\Output{The subset $\mathcal{S}_c \subseteq \mathcal{S}$ with $|\mathcal{S}_c|=k$ and corresponding weight $\{\gamma_s\}_{s=1}^k$. }
\BlankLine
\For{$i \leftarrow 1$ \KwTo $n_c$}
{
$\mathcal{S}_c^{(i)} \leftarrow \emptyset$; \tcc{initialize the coreset for each class}
}
\For{$i \leftarrow 1$ \KwTo $n_c$}
{
$\mathcal{S}_c^{(i)} \leftarrow \bm{x} \in \mathcal{S}^{(i)}$\; \tcc{randomly select point $\bm{x}$ from set $\mathcal{S}^{(i)}$ associated with class $i$, and add it to set $\mathcal{S}_c^{(i)}$.}
\While{$|\mathcal{S}_c^{(i)}| < \ceil*{\frac{|\mathcal{S}^{(i)}|}{|\mathcal{S}|}\cdot k}$}
{
	$\bm{x} = \argmax_{\bm{x}\in\mathcal{S}^{(i)}\textbackslash 	{\mathcal{S}_c^{(i)}}} \min_{\bm{x}'} d(\bm{x},\bm{x}')$\;
	$\mathcal{S}_c^{(i)} \leftarrow \mathcal{S}_c^{(i)} \bigcup \{\bm{x}\} $\;
}
}
$s = 1$\;
\For{$j \leftarrow 1$ \KwTo $n_c$}
{
    \For{ $\bm{x} \in \mathcal{S}_c^{(j)}$}
    {
    $\gamma_{s} \leftarrow \sum_{\bm{x}'\in \mathcal{S}^{(j)}}\bm{I}_{|\bm{x}-\bm{x}'|\leq \delta_c^{(j)}}$\;
	\tcc{use the indicator function to count the number of same-class training samples in radius $\delta_c$.}
	$s\leftarrow s+1$\;
    }
}

$\mathcal{S}_c \leftarrow \emptyset$\;
\For{$i \leftarrow 1$ \KwTo $n_c$}
{
	$\mathcal{S}_c \leftarrow \mathcal{S}_c \bigcup \mathcal{S}_c^{(i)}$
}
return $\mathcal{S}_c$ and $\{\gamma_s\}_{s=1}^k$.
\caption{The greedy algorithm for $k$-center coreset selection.}\label{alg:coreset}\DecMargin{1em}
\end{algorithm}\DecMargin{1em}
\vspace{1em}

For each class $i\in [n_c]$, at the first step, we randomly pick a data point $\bm{x}$ from the set $\mathcal{S}^{(i)}\subseteq \mathcal{S}$ and put it into the initialized empty set $\mathcal{S}_c^{(i)}$, where $\mathcal{S}^{(i)}$ refers to the set of all training data points associated with the label $i$. In the second step, we iteratively choose the data point from $\mathcal{S}^{(i)}$ to be as far away as possible from the other centers in $\mathcal{S}_c^{(i)}$. That is, for each class $i$, repeatedly finding a data point $\bm{x}$ for which the distance $d(\bm{x},\bm{x}')$ is maximized where $ \bm{x}\in \mathcal{S}^{(i)},\bm{x}' \in \mathcal{S}_c^{(i)}$. The searched data point is then appended to $\mathcal{S}_c^{(i)}$. This iteration procedure is terminated  when $|\mathcal{S}_c^{(i)}| \geq \ceil*{\frac{|\mathcal{S}^{(i)}|}{|\mathcal{S}|}\cdot k}$. In other words, the coreset size of each class $i$ is restricted to be proportional to the ratio of the size of that class $|\mathcal{S}^{(i)}|$ to the overall amount of data $|\mathcal{S}|$. 
When the coreset for each class is collected, we merge these sets together to create a coreset with $\mathcal{S}_c = \cup_{i=1}^{n_c} \{\mathcal{S}_c^{(i)}\}$ and set the weight $\gamma_s$ as the number of samples covered by the coreset example $\bm{x}_s$ in radius $\delta_c$. Once the coreset $\mathcal{S}_c$ is built, we can integrate it into QML models. In the following, we introduce how it works on QNNs and quantum kernels. 

The Pseudo code of QNN trained over the coreset is summarized in Algorithm~\ref{alg:qnn_coreset}. Conceptually, we only need to replace the full training data $\mathcal{S}$ with the coreset $\mathcal{S}_c$, and rewrite the cost function by introducing the weight $\{\gamma_s\}$ for each corresponding data point $\bm{x}$ in coreset $\mathcal{S}_c$, i.e.,
\begin{equation}
R_{c}^{\text{QNN}} = \frac{1}{|\mathcal{S}_c|}\sum_{\bm{x}_s,y_s\in\mathcal{S}_c}\gamma_s \cdot l(tr[M\cdot |\bm{x}_s,\theta\rangle \langle\bm{x}_s,\theta|,y_s),
\end{equation}
where $|\bm{x}_s,\theta\rangle$ and $M$ are identical to those in Eq. (\ref{eqn:sec:QNN-preliminary}). 

\vspace{1em}
\IncMargin{1em}
\begin{algorithm}[H]
\SetKwData{Left}{left}\SetKwData{This}{this}\SetKwData{Up}{up}
\SetKwFunction{Union}{Union}\SetKwFunction{FindCompress}{FindCompress}
\SetKwInOut{Input}{input}\SetKwInOut{Output}{output}

\Input{Coreset $\mathcal{S}_c$, weights $\{\gamma_s\}_{s=1}^k$, maximum epoch number $K_{max}$, learning rate $\eta$}
\Output{trained QNN.}
\BlankLine
Initialize the parametrized quantum circuit $U(\theta)$\;
$l = 1$\;
\While{$l < K_{max}$}
{
\For{$\bm{x}_j,y_j \in \mathcal{S}_c$}
{	
	Estimate the gradient of $R_c$ with respect to $\theta$, $\nabla_{\theta} f(\bm{x}_j,y_j)$ \;
	$\theta_{l+1} \leftarrow \theta_{l} + \eta \cdot \gamma_j\cdot \nabla_{\theta} f(\bm{x}_j,y_j)$;
}
}
return the QNN $tr[M\cdot U(\theta)U(\bm{x}_i)\rho_0 U^\dagger(\theta) U^\dagger(\bm{x}_i)]$.
\caption{QNN with coreset selection}\label{alg:qnn_coreset}\DecMargin{1em}
\end{algorithm}\DecMargin{1em}
\vspace{1em}

We next explain the implementation of a quantum-kernel-based support vector machine (SVM) classifier over coreset. As we employ the coreset $\mathcal{S}_c$ and introduce the weights $\{\gamma\}$, there is a slight difference between original SVM and coreset enhanced SVM, where the Lagrange multipliers $\alpha_i$ in dual problem is upper bounded by $C\cdot \gamma_i$ instead of $C$. Mathematically, we have
\begin{align}\label{eq:dual_svm_coreset}
\max_{\alpha} \quad & \sum_i\alpha_i - \frac{1}{2}\sum_{i,j}\alpha_i\alpha_j y_i y_j K(\bm{x}_i,\bm{x}_j),\\
\textrm{s.t.} \quad & \sum_{i}y_{i}\alpha_{i} = 0, \nonumber \\
 &0\leq \alpha_i \leq C\cdot \gamma_i, i=1,\cdots,k. \nonumber
\end{align}
where $K(\bm{x},\bm{x}')=|\langle \bm{x}|\bm{x}'\rangle|^2$ is the quantum kernel function. The training process is similar to the original SVM and depicted in Algorithm~\ref{alg:qkernel_coreset}.

\vspace{1em}
\IncMargin{1em}
\begin{algorithm}[H]
\SetKwData{Left}{left}\SetKwData{This}{this}\SetKwData{Up}{up}
\SetKwFunction{Union}{Union}\SetKwFunction{FindCompress}{FindCompress}
\SetKwInOut{Input}{input}\SetKwInOut{Output}{output}

\Input{Coreset $\mathcal{S}_c$, weights $\{\gamma_s\}_{s=1}^k$, encoding circuit $U$, regularization parameter $C$}
\Output{SVM with quantum kernel.}
\BlankLine
\For{$\bm{x} \in \mathcal{S}_c$}
{	
	\For{$\bm{x}' \in \mathcal{S}_c$}
	{	
		$K_{\bm{x},\bm{x}'} = |\langle 0|U^\dagger(\bm{x}')|U(\bm{x})|0\rangle|^2$\;
		\tcc{construct kernel matrix $K$}
	}
}
Solve the dual problem in Eq.(\ref{eq:dual_svm_coreset}), get the optimal parameters $\alpha^*$\;
Calculate the parameter $W^*,b^*$;
\begin{equation}
W^* = \sum_i \alpha_i y_i \bm{x}_i,
\end{equation}
\begin{equation}
b^* = y_i - \sum_i \alpha_i y_i K(\bm{x}_i,\bm{x}_j),
\end{equation}
return the quantum kernel based SVM, $f(\bm{x})=\sum_i\alpha_iy_iK(\bm{x},\bm{x}_i)+b^*$.
\caption{Quantum kernel with coreset selection.}\label{alg:qkernel_coreset}\DecMargin{1em}
\end{algorithm}\DecMargin{1em}
\vspace{1em}

\section{Generalization ability of QML models under coreset selection}

In machine learning, generalization analysis is a crucial theoretical tool to measure how accurately a learning model predicts previously unseen data \cite{mohri2018foundations}. It helps us as a guiding metric to choose the best-performing model when comparing different model variations. As explained in Section \ref{sec:pred-ML}, the purpose of learning algorithms is to find a hypothesis $f$ such that the true risk $R$ on the data distribution $\mathcal{Z}$ is as small as possible. However, it is unlikely to directly estimate the true risk $R$ since the distribution $\mathcal{Z}$ is unknown and an alternative solution is minimizing the empirical risk $R_e$ over the given samples. The generalization error quantifies the gap between the true risk $R$ and empirical risk $R_e$, i.e.,
\begin{equation}
|R - R_e|= \left|\mathbb{E}_{{({{\bm{x}}},y)}\sim p_{\cal{Z}}}[l(f_{\bm{w}}(\bm{x}),y)] - \frac{1}{|\mathcal{S}|}\sum_{({{\bm{x}}}_i,y_i)\in {\cal{S}}_t}l(f_{\bm{w}}({{\bm{x}}}_i) ,y_i)\right|.
\end{equation}
Thus, the true risk $R$ can be upper bounded by its empirical error and the generalization error, i.e., 
\begin{equation}
    R \leq \underbrace{\left|R - R_e\right|}_{\text{generalization error}} +\underbrace{\left|R_e\right|}_{\text{empirical error}}.
\end{equation}



Since the empirical error, namely the training loss, is close to zero in general, we can only consider the upper bound of the generalization error. 
Thus, a natural question is whether the coreset selection could provide a tighter generalization bound compared to random sampling under the same pruned training size. To answer this question, we need first define the empirical error $R_r$ on a sub-training set $\mathcal{S}_r$ that is generated by random sampling from full data set $\mathcal{S}$ as, 
\begin{equation}
    R_r = \frac{1}{|\mathcal{S}_r|}\sum_{\bm{x}_i,y_i\in\mathcal{S}_r}l(f_{\bm{w}}(\bm{x}_i),y_i),
\end{equation}
where $|\mathcal{S}_r|=N_r$ is the size of $\mathcal{S}_r$. Hence, the generalization error $G_r$ over $\mathcal{S}_r$ is represented as,
\begin{equation}
    G_r = \left|R - R_r\right|=\left|\mathbb{E}_{{({{\bm{x}}},y)}\sim p_{\cal{Z}}}[l(f_{\bm{w}}(\bm{x}),y)] - \frac{1}{|\mathcal{S}_r|}\sum_{\bm{x}_i,y_i\in\mathcal{S}_r}l(f_{\bm{w}}(\bm{x}_i),y_i)\right|.
\end{equation}
According to previous studies on the generalization analysis of QNNs and quantum kernels \cite{caro2022generalization, Huang2021PowerOD}, with probability at least $\ge 1-\delta$ over $\mathcal{S}_r$, we have the generalization error of QNNs,
\begin{equation}
    G_r^{\text{QNN}} \leq \mathcal{O}\left(\sqrt{\frac{m\log{(m)}}{N_r}}+\sqrt{\frac{\log{(1/\delta)}}{N_r}}\right),
\end{equation}
where $m$ is the trainable parameters of QNNs, and the generalization error of quantum kernel,
\begin{equation}
    G_r^{\text{qkernel}} \leq \mathcal{O}\left(\sqrt{\frac{\ceil{\|\bm{w}\|}^2}{N_r}}+\sqrt{\frac{\log{(4/\delta)}}{N_r}}\right).
\end{equation}
To bound the generalization error $G_c$ on coreset $\mathcal{S}_c$, we similarly define the empirical error $R_c$ on the coreset $\mathcal{S}_c$ as,
\begin{equation}\label{eq:gamma}
    R_c = \frac{1}{|\mathcal{S}_c|}\sum_{\bm{x}_s,y_s\in\mathcal{S}_r}\gamma_s\cdot l(f_{\bm{w}}(\bm{x}_s),y_s),
\end{equation}
where $|\mathcal{S}_c|=N_c$ is the size of $\mathcal{S}_c$ and $\gamma_s$ is the weight for each coreset example $\bm{x}_s$ that is used to make $R_c$ approximate to $R_e$ over full data set $\mathcal{S}$. Then, we can present $G_c$ as,
\begin{equation}
    G_c = |R - R_c|=\left|\mathbb{E}_{{({{\bm{x}}},y)}\sim p_{\cal{Z}}}[l(f_{\bm{w}}(\bm{x}),y)] - \frac{1}{|\mathcal{S}_c|}\sum_{\bm{x}_s,y_s\in\mathcal{S}_c}l(f_{\bm{w}}(\bm{x}_s),y_s)\right|.
\end{equation}
Based on the triangle inequality, $G_c$ can be bounded by,
\begin{equation}
    G_c = \left|R - R_c\right|\leq \underbrace{|R - R_e|}_{\text{generalization error}} + \underbrace{|R_e - R_c|}_{\text{coreset error}}.
\end{equation}
Here, we name the second term, i.e. $|R_e - R_c|$, as the coreset error that describes the gap between the empirical loss on the full data set and the coreset.   As the bound of the generalization error term, i.e. $|R - R_e|$, is given by previous works \cite{caro2022generalization,Huang2021PowerOD}, we can only focus on the bound of coreset error. 

By analyzing the coreset error, we first provide the generalization error bound of QNNs. 
\begin{theorem}\label{thm:qnn-coreset}{\bf{(Generalization error bound of QNN on coreset)}}
Given sample set $\mathcal{S} = \{{{\bm{x}}}_i,y_i\}_{i=1}^{N_t}$ are $i.i.d$ randomly drawn from distribution $\mathcal{Z}$, $\mathcal{S}_c$ is $\delta_c$ cover of $\mathcal{S}$. Assume
there is $\lambda_\eta$-Lipshcitz continuous class-specific regression function $\eta(\bm{x}) = p(y = c|\bm{x})$, and the loss $l(f_{\bm{w}}(\bm{x}_s,y_s))$ over coreset $\mathcal{S}_c$ is zero and bounded by $L$. We have the following upper bound for the generalization error of QNNs trained on coreset with probability $1-\delta$,
\begin{equation}\label{eq:gen_qnn}
G_c^{\text{QNN}}\leq \mathcal{O}\left(\sqrt{\frac{m \log(m)}{N_t}}+\sqrt{\frac{\log(1/\delta)}{N_t}}+\delta_c (\lambda_\eta L n_c + d\sqrt{d_{\bm{x}}}\max_j|\bm{w}_j|\left|M\right|(\left|M\right|+\max|y|)\right),
\end{equation}
where $m$ is the number of the parameters in QNN, $d$ is the number of QNN layers, $n_c$ is the number of classes, $M$ is the measurement operator, $d_{\bm{x}}$ is the feature dimension. 
\end{theorem}
For the coreset error $|R_e - R_c|$, we assume the training error on coreset is equal to zero, thus it becomes the average error over the entire data set which can be bound with the radius $\delta_c$ determined by the $k$-center covering problem shown in Fig. \ref{fig:kcp}, which is related to the data prune rate $\zeta=|\mathcal{S}_c|/|\mathcal{S}|$. Combining the generalization bound on the full data set and the bound of risk gap between the full data set and coreset together, we have the generalization error of QNN on coreset is mainly bounded by two terms, i.e. $\mathcal{O}((\sqrt{m \log(m)/N_t}+\delta_c)$. Thus, the $G_c^{\text{QNN}}$ will give a tighter bound than $G_r^{\text{QNN}}$ when we carefully choose the data prune rate $\zeta$. It is clear that $G_c^{\text{QNN}}$ gives the tighter bound on the first term since $N_r < N_t$. For the second term, as $\delta_c$ is related to $\zeta$, a low $\zeta$ will cause $\delta_c$ to become large leading to a high approximation error. Conversely, a high $\zeta$ will decrease the approximation error but the acceleration of training will disappear because $N_c$ is approximate to $N_t$.

We next provide the generalization error bound of quantum kernels under the coreset. 
\begin{theorem}\label{thm:qkernel-coreset}{\bf{(Generalization error bound of quantum kernels on coreset.)}}
Given sample set $\mathcal{S} = \{\bm{x}_i,y_i\}_{i=1}^{N_t}$ are $i.i.d$ randomly drawn from distribution $\mathcal{Z}$, $\mathcal{S}_c$ is $\delta_c$ cover of $\mathcal{S}$. Assume there is $\lambda_\eta$-Lipshcitz continuous class-specific regression function $\eta(\bm{x}) = p(y = c|\bm{x})$, and the loss $l(f_{\bm{w}}(\bm{x}_s,y_s))$ over coreset $\mathcal{S}_c$ is zero and bounded by $L$. We have the following upper bound for the generalization error of SVM with quantum kernel trained on coreset with probability $1-\delta$,
\begin{equation}\label{eq:gen_qkernel}
G_c^{\text{qkernel}}\leq \mathcal{O}\left(\sqrt{\frac{\ceil{\|\bm{w}\|}^2}{N_t}}+\sqrt{\frac{\log{(4/\delta)}}{N_t}}+\delta_c (\lambda_\eta L n_c + N_c\sqrt{d_{\bm{x}}}\max_j|\bm{w}_j|\cdot(1+(N_q-1)r))\right),
\end{equation}
where $n_c$ is the number of classes, $d_{\bm{x}}$ is the feature dimension of $\bm{x}$, $N_q$ is the size of mapped quantum state $|\bm{x}\rangle$, r is the maximum value of feature $\bm{x}$. 
\end{theorem}
Similar to the analysis of the generalization error of QNNs on coreset, the generalization error of SVM with a quantum kernel is also mainly bounded by two terms, i.e. $\mathcal{O}(\sqrt{\frac{\ceil{\|\bm{w}\|}^2}{N_t}}+\delta_c)$ that indicates we have the similar results like $G_c^{\text{QNN}}$. Here, when we employ coreset selection to reduce the size of training examples, we will reduce the complexity of getting kernel matrix, that provides $\mathcal{O}(k^2)$ speedup where $k 
 = N_t/N_c$, meanwhile also having provable generalization guarantee.\\
It should be pointed out that it is not easy to propose a universal trick for selecting an appropriate data prune rate $\zeta$ that determines the $\delta_c$ to achieve a good generalization performance. Because it is related to the distribution of the given training set $\mathcal{S}$ and unknown true data distribution $\mathcal{Z}$. Nevertheless, we give some numerical experiments on various data and models, the results not only support the analytical findings but also might provide some practical advice for the selection of $\zeta$. Consequently, if we carefully choose the data prune rate and the size of the training set to be scaled at least quasi-linearly with the number of gates employed, we are able to accelerate the quantum machine learning model with a performance guarantee. These results provide effective and practical guidance on achieving accurate and reliable performance with a reasonable model complexity and sample complexity.




\section{Numerical results}
In this section, we conduct extensive numerical simulations to explore the performance of the proposed coreset selection method. Specifically, we employ our proposal to accomplish three learning tasks, which are synthetic data classification, quantum correlation identification, and quantum compiling. 

\subsection{Synthetic data classification by quantum kernels}
We first utilize the quantum kernel to classify a synthetic dataset with coreset. The construction rule of the synthetic dataset mainly follows Ref.~\cite{Huang2021PowerOD}. Specifically, given the data set $\{{\bm{x}}_i,y_i\}_{i=1}^N$ that independently sampled from the distribution $\mathcal{X}$, the corresponding labels are modified according to the maximized geometric difference given by
\begin{equation}\label{eq:geom}
\max_{y\in\mathbb{R}^N}\frac{\sum_{i=1}^N\sum_{j=1}^N (K^{Q})^{-1}_{ij}y_iy_j}{\sum_{i=1}^N\sum_{j=1}^N (K^{C})^{-1}_{ij}y_iy_j},
\end{equation}
where $K^Q$ and $K^C$ denote the quantum and classical kernel respectively. The optimal solution of Eq. (\ref{eq:geom}) yields the modified labels ${\bf{y}}^*$ such that maximize the geometric difference.
\begin{equation}
{\bf{y}}^*= \text{sign}(\sqrt{K^Q}v),
\end{equation}
where $v$ is the eigenvector of $\sqrt{K^Q}(K^C)^{-1}\sqrt{K^Q}$ with the maximum eigenvalue, and $\text{sign}(\bf{z})$ is the element-wise function such that set the $i$-th element as $+1$ if $z_i > median({\bf{z}})$ otherwise set as $-1$. As proved in Ref.~\cite{Huang2021PowerOD}, quantum kernels can achieve quantum advantages when learning this dataset.

\begin{figure}[H]
  \centering
  \includegraphics[width=15cm,height=5cm]{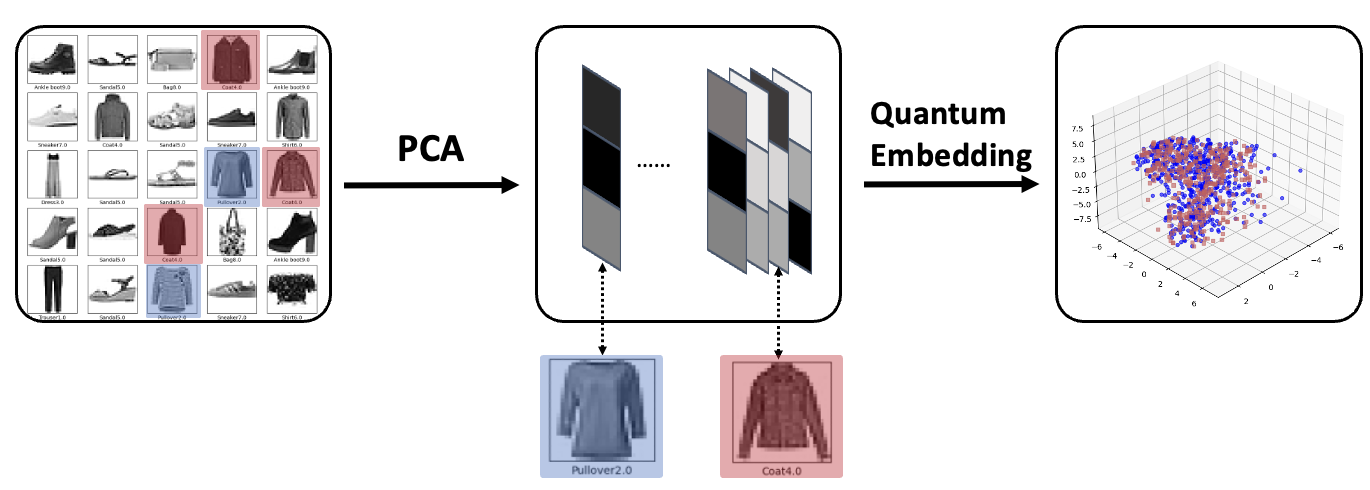}
  \caption{Illustration of the employed synthetic dataset adapted from the fashion-MINIST \cite{xiao2017fashion}. We use principal component analysis to get the low-dimensional representation, then embed the reduced data into quantum Hilbert space. In the end, we relabel the data according to maximizing the geometric difference between classical and quantum kernels in Eq.~(\ref{eq:geom}).}
  \label{fig:synthetic}
\end{figure}

An illustration of the synthetic dataset construction is shown in Fig.~\ref{fig:synthetic}. Concretely, in our numerical simulations, the synthetic dataset is based on fashion-MNIST \cite{xiao2017fashion}. As the dimension of vectorized data of the fashion-MNIST is too high for NISQ device, we preprocess it as the low-dimension representation by principle component analysis and then relabel the class of each data point according to Eq.~(\ref{eq:geom}). Besides, the element of classical kernel $K^C_{ij}$ is given by the radial basis function kernel $K^C_{ij}=\exp(-\frac{||{\bm{x}}_i - {\bm{x}}_j||^2}{2\sigma^2})$. The quantum kernel $K^Q_{ij}$ is generated through encoding the data points into $N_q$-qubit Hilbert space by a quantum circuit $U_e$,
\begin{equation}
K^Q_{ij} = tr(\rho({\bm{x}}_i)\rho({\bm{x}}_j)),
\end{equation}
where $\rho({\bm{x}}) = U_e({\bm{x}})|0\rangle\langle0|U^{\dagger}_e({\bm{x}})$. We can further assume the following form of $U_{{\bm{x}}}$ implemented through quantum gates in an $N_q$-qubit circuit: $U_{{\bm{x}}}=(U({{\bm{x}}})H^{\otimes N_q})^2 |0\rangle^{\otimes N_q}$, where $U({{\bm{x}}})=\exp(\sum_{j=1}^{N_q}{{\bm{x}}}_j\sigma^Z_j+\sum_{j,j'=1}^{N_q}{{\bm{x}}}_j{{\bm{x}}}_{j'}\sigma^Z_j\sigma^Z_{j'})$.

\begin{figure}[H]  
  \centering
  \includegraphics[width=15cm,height=4cm]{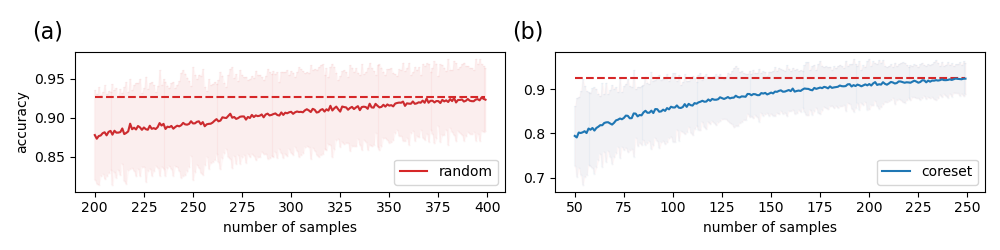}
  \caption{The comparison between the proposed model with random sampling and coreset selection. It shows the average performance of classification of synthetic data where the solid line is the test accuracy of the models that trained over coreset and random sampling. The shaded area refers to the range of the test accuracy. The red dashed line denotes the maximum accuracy achieved by random sampling.}
  \label{fig:qkernel}
\end{figure}

Once the synthetic data set was prepared, we conducted experiments on training sets with various sizes and subsequently tested on 200 unseen examples to get the test accuracy.  Instead of independently and randomly choosing training data from entire set $\mathcal{S}$, we first solve the $k$-center problem over set $\mathcal{S}$ with $1000$ examples that are equivalent to form the coreset $\mathcal{S}_c$. 
In Fig.~\ref{fig:qkernel}, we show the comparison of the classification performance under these settings. They respectively depict the correlation between the size of the training set, obtained using random sampling and coreset selection, and the corresponding test accuracy. For experiments involving the same number of training examples, we carried out five independent trials, each employing randomized initialization, resulting in the shaded area as shown in Fig.~\ref{fig:qkernel}. The average test performance of these trials is plotted as a solid line. We also highlight the best test accuracy of random sampling as the red dash line. From the results shown in Fig.~\ref{fig:qkernel}, when the models are trained on a set with 250 examples (i.e., $\zeta = 0.25$), the test accuracy of those trained on the data points obtained by coreset selection is higher than that of the model via random sampling. Besides, when achieving the same test accuracy, the model only needs 250 coreset examples while requiring 400 randomly sampled examples. They support our analytical results, the QNNs trained on coreset have better generalization performance than those with randomly picked training data under appropriate data prune rate. The coreset-enhanced classifier significantly improves the training efficiency, achieving competitive performance while only utilizing approximately $50\%$ of the training examples compared to random sampling.

\subsection{Correlation identification by QNNs}\label{sec:quantum correlation}
Non-classical correlation plays a core role in quantum information and quantum computation \cite{gross2009most}. 
Nevertheless, identifying the non-classical correlation of a given quantum state is a challenging task. Ref.~\cite{yang2019experimental} explored classifying non-classical correlation experimentally with machine learning techniques.  Consider the family of quantum states characterized by $p$ and $\theta$ with the following form,
\begin{equation}
\rho_{AB}(p,\theta) = p|\psi_\theta\rangle\langle\psi_\theta| + (1-p)\frac{\mathbb{I}}{2}\otimes tr_A(|\psi_\theta\rangle\langle\psi_\theta|),
\end{equation}
where the $p\in (0,1)$, $\theta \in (0,2\pi)$ and state $|\psi_\theta\rangle = \cos(\theta)|00\rangle + \sin(\theta)|11\rangle$. There are following rules to determine the non-classical correlation of quantum states $\rho_{AB}$ including separable, entangled, one-way steerable, and non-local. 
\begin{enumerate}

\item According to PPT criterion, the states are \textit{separable} when $p<\frac{1}{3}$ otherwise they are entangled.

\item When $\frac{1}{\sqrt{2}} < p < \frac{1}{\sqrt{1+\sin^2(2\theta)}}$, the quantum state is \textit{one-way steerable}. 

\item When $p > \frac{1}{\sqrt{1+\sin^2(2\theta)}}$, the state is \textit{non-local}.

\end{enumerate}
Therefore, to identify the non-local correlation of a given state under the learning framework, we can label the quantum states with different non-classical correlations according to the above criteria and create the data set $\{\rho^{AB}_{j}, y_j\}_{j=1}^N$ where $y_i$ represents the type of correlation, i.e. separable, entangled, one-way steerable, non-local.

\begin{figure*}
    \centering
\subfloat[quantum states]
          {\includegraphics[width=0.3\textwidth]{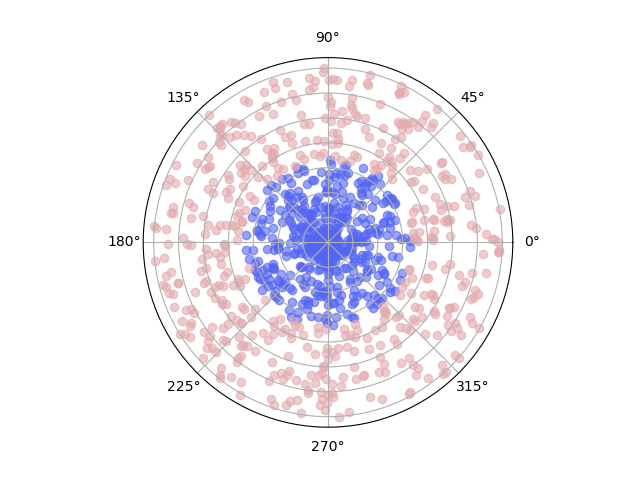}
          \label{fig:correlation}}
\subfloat[test accuracy of random sampling and coreset ]
		{\includegraphics[width=0.7\textwidth]{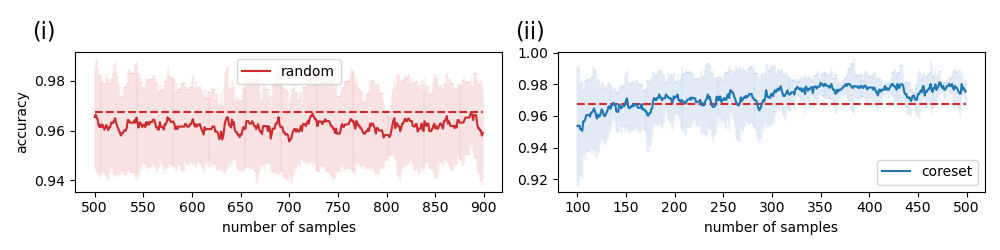}
		\label{fig:qnn_correlation}}

\caption{\small{ Results related to correlation identification. (a) The states are represented as dots in polar plots. The radius of the polar plot represents the parameter $p$ that varies from $0$ to $1$. The phase stands for the parameter $\theta$ that varies from $0$ to $2\pi$. The blue and orange dots indicate the separable and
entangled states separately. (b) The comparison of test accuracy between the classifier with random sampling and coreset under different sizes of training set. }}
\label{fig:qnn_correlation}
\end{figure*}

To further enhance the methods, the fewer training samples the less run time, we apply the coreset selection to this learning task. For classifying the quantum correlation, we uniformly pick the parameters $p \in (0,1)$ and $\theta \in (0,2\pi)$ to generate the 1000 quantum states as the full training set as shown in the subfigure (a) of Fig.~\ref{fig:qnn_correlation}. Then label the class of correlation in terms of the criterion listed above. Here, we continue with the same experimental setup as the previous section. 
where we examine the test accuracy of the model, trained over various sample sizes, on 200 unseen random samples. 

The red dash lines in subfigure (i) and (ii) denote the maximum test accuracy attained by the classifier through random sampling, with a maximum of 900 training samples. When utilizing the coreset method to prune the dataset, for sample sizes exceeding 180, the averaged test accuracy obtained by coreset is almost higher than the best accuracy achieved by random sampling.  It indicates that when $\zeta \ge 0.18$, the coreset selection could provide better performance compared to training on random sampling data, which matches our theoretical findings.

\subsection{Quantum compiling by QNN}
Compiling a unitary into a sequence of gates is a challenging and high-profile task for NISQ quantum devices. With the limitation of current NISQ hardware, compiling a unitary not only should take account of the function but also consider the connectivity and depth of the output circuit. Recently, various methods for quantum compiling have been proposed under the framework of variational quantum algorithms \cite{khatri2019quantum, caro2022generalization, bilkis2021semi}. In general, these algorithms consider the compiling task as an optimization problem on a given compact quantum gates set which consists of fixed and parametrized quantum gates. The goal is to optimize the structure and gate parameters such that the proposed quantum circuit approximates to the given unitary. 

Here, we consider the method that tackles the compiling task by a quantum machine learning protocol. Given an $n$-qubit target unitary $U$, the training data consists of random input states and their corresponding output when $U$ applied on, i.e. $\{|\psi_j\rangle, U|\psi_j\rangle\}_{j=1}^{N}$. To approximate the target unitary $U$, one can simply minimize the empirical loss of the squared trace distance between target states $U|\psi_j\rangle$ and parametrized output states $V(\bm{\theta})|\psi_j\rangle$ and over randomly sampled states in Hilbert space. Concretely, we randomly pick the quantum gates from the gate pool consisting of single parametrized gates $\{R_x, R_y, R_z\}$ and CNOT gate to build the target quantum circuits $U$. Then set the input as the random state $|\psi_i\rangle$ in Hilbert space and get training pairs $\{|\psi_j\rangle,U|\psi_j\rangle\}_{j=1}^{1000}$. In subfigure (b) of Fig.~\ref{fig:qnn_compile}, we present the comparison between the models trained on random sampling and coreset with various data sizes and pruning ratios. We separately estimate the performance of the model with $\zeta=\{0.8, 0.4, 0.6, 0.2\}$. Each data point with a different color on the plot corresponds to different $\zeta$. For the compiling task on a specific unitary, we found that 1000 training examples are redundant, and the effective size of training data scales linearly with the system size which is similar to what previous work found \cite{caro2022generalization}. When given around 20 examples, it is sufficient for training the variational compiler to achieve a reasonable performance for a 6-qubit system. Since the randomly sampled training points are likely to be uniformly located in Hilbert space and the proposed coreset selecting approach also uses $k$-centers to uniformly cover the entire set, the compilers trained on coreset with the different $\zeta$ have similar performance compared to those trained on the random-sampled set.

\begin{figure*}
    \centering
\subfloat[target ansatz]
          {\includegraphics[width=0.4\textwidth]{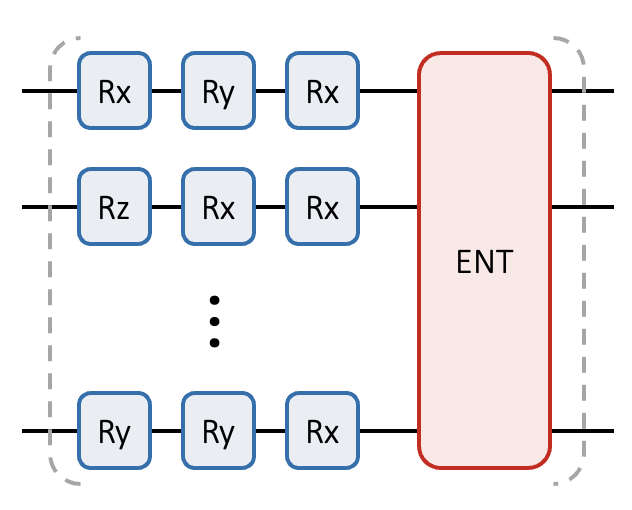}}
\subfloat[performance of random sampling and coreset ]
		{\includegraphics[width=0.6\textwidth]{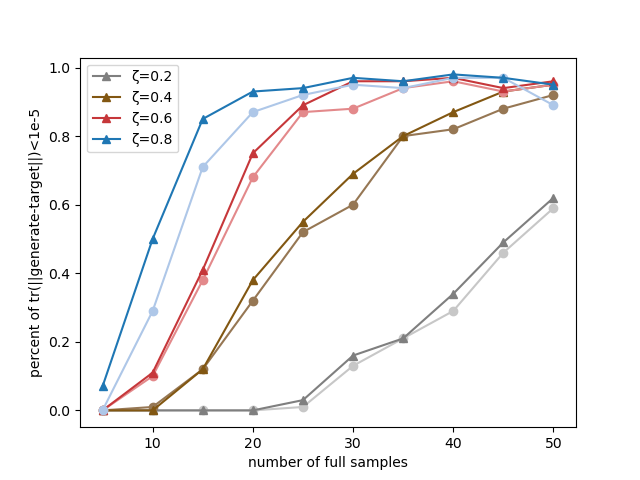}}

\caption{\small{Results of unitary compiling. (a) The target unitary is exploited in numerical simulations. (b) The performance comparison between the different prune ratios $\zeta$ of coreset on compiling task. The solid dark and light lines represent the models trained on the coreset and random samples. The vertical axis represents the percentage of states in the test set for which the trace distance to their corresponding targets is below $10^{-5}$. }}
\label{fig:qnn_compile}
\end{figure*}

\section{Discussion}

In this work, we investigate enhancing QML model from the data engineering perspective and attempt to alleviate a practical problem when handling a large volume of data samples. In particular, we consider the coreset construction as a $k$-set cover problem and then analyze the generalization performance of QML models on coreset. Our investigation of various learning scenarios, including on classification of synthetic data, identification of non-classical correlations in quantum states, and quantum circuit compilation, confirm the extreme improvements in the effectiveness of our proposal. 

Our research findings highlight the considerable enhancement in model training achieved through the utilization of the coreset method. Data pruning contributes to improved training efficiency. Besides, it also helps filter out noise data, thereby enhancing model performance. 
It's evident that selecting a sparser coreset enforces a more rigorous upper bound on the count of trainable gates and appropriate data prune rate. The size of the training set should be scaled at least quasi-linearly with the number of gates. These findings provide effective and practical guidelines to achieve accurate and reliable results with a reasonable configuration of gates and training data. Although we have improved model training efficiency through data pruning methods, there is still significant room for enhancement. For instance, the introduction of influence functions could better characterize the impact of data variations on the model, thus achieving more precise data filtering. 

We aspire for the work presented in this article to offer valuable insights and guidance for future practical research in quantum machine learning, from theoretical or practical aspects, on NISQ devices.

\bibliographystyle{unsrt}
\bibliography{mybib.bib}
%
%
%
%

\onecolumn\newpage
\appendix

\section{Appendix}

\begin{theorem}\label{thm:qnn-generalization}{\bf{(Generalization bound for QNN \cite{caro2022generalization})}} 
Let the QNN to be trained consists of $m$ trainable parameterized 2-qubit or 1-qubit gates, and an arbitrary number of non-trainable fixed gates. Suppose that, given training data ${\cal{S}}=\{{{\bm{x}}}_i, y_i\}_{i=1}^{N_t}$ and a loss function $l(f_{\bm{\theta}};{{\bm{x}}}_i,y_i)$ bounded by $L$, our optimization yields optimized parameters $\boldsymbol{\theta^{*}}$. Then, according to Theorem C.6 in \cite{caro2022generalization}, the generalization error of the optimized QNN is bounded by the following form with probability at least $(1-\delta)$ over the choice of $i.i.d$ training data ${\cal{S}}_t$ from $\mathcal{Z}$:
\begin{equation}
\left|R-R_{e}\right| < \mathcal{O}\left(\sqrt{\frac{m \log(m)}{N_t}}+\sqrt{\frac{\log(1/\delta)}{N_t}}\right).
\end{equation}
\end{theorem}

\begin{lemma}\label{lemma:qnn}
The loss function $l(f_{\bm{\theta}}(\bm{x}),y)$ of $d$-layer QNN introduced in section \ref{sec:qnn} is the $\lambda$-Lipschitz in feature space with 
\begin{equation}
\lambda=2d\sqrt{d_{\bm{x}}}\max_j|\bm{w}_j|\left|M\right|(\left|M\right|+\max|y|).
\end{equation}
\end{lemma}
\begin{proof}
Assume given data $\bm{x}=(x_1,\cdots,x_{d_{\bm{x}}})$, data encoding scheme used by the QNN is the $d$-layer re-uploading scheme, the mapped quantum state is written as,
\begin{equation}
\ket{\bm{x},\bm{\theta}}=U({\bm{x}},\bm{\theta}^{(d)},\bm{w}^{(d)})...U({\bm{x}},\bm{\theta}^{(1)},\bm{w}^{(1)})\ket{0},
\end{equation}
where $\bm{w}$ is encoding parameters in data-reuploading, the $i$-th layer $U({\bm{x}},\bm{\theta}^{(i)},\bm{w}^{(i)})$ can be presented as,
\begin{equation}
    U({\bm{x}},\bm{\theta}^{(i)},\bm{w}^{(i)})= U(\bm{\theta}^{(i)})\prod_k\exp{(i w_k^{(i)} x_k P_k)}, 
\end{equation}
where $P$ is a Pauli gate and $\bm{\theta},\bm{w}$ both trainable parameters.

Assume the Jacobian vector $J$ of the loss function $l=(\langle \bm{x},\bm{\theta}|M|\bm{x},\bm{\theta}\rangle - y)^2$ is bounded in the feature space, now we investigate the Lipschitz constant of $l$.
\begin{equation}
    \left|l({\bm{x}})-l(\bm{x}')\right| \leq \max \|J\|_2 \|{\bm{x}}-\bm{x}'\|_2, \forall {\bm{x}}, \bm{x}' \in \mathbb{R}^{d_{\bm{x}}},
\end{equation}
where the Jacobian matrix is,
\begin{equation}
    J = \left[ \frac{\partial l}{\partial x_1},\cdots,\frac{\partial l}{\partial x_{d_{\bm{x}}}} \right].
\end{equation}
Thus the maximum norm of the Jacobian matrix $J$ can be bounded by the maximum norm of derivative of $l$ with respect to $\bm{x}$,
\begin{equation}
    \max\|J\|_2 \leq \sqrt{d_{\bm{x}}}\max_j\left|\frac{\partial l}{\partial x_j}\right|,
\end{equation}
where the $\left|\frac{\partial l}{\partial x_j}\right|$ can be bounded by
\begin{align}
    \left|\frac{\partial l}{\partial x_j}\right| &= \left|\left( \bra{\psi}M\ket{\psi}-y\right)\cdot 2\mathcal{R}\left(\bra{\psi}M\frac{\partial \ket{\psi}}{\partial x_j}\right) \right|
    \cr
    &\leq (\left|M\right|+\max|y|)\cdot 2|M|\cdot \left|\frac{\partial \ket{\psi}}{\partial x_j}\right|_2.
\end{align}
Here, $\left|M\right|$ denotes the spectral norm of the measurement operator.\\
To bound $\|\frac{\partial \ket{\psi}}{\partial x_j}\|_2$, let $U^{(k)}=U({\bm{x}},\bm{\theta}^{(k)},\bm{w}^{(k)})$, we first apply product rule to $\frac{\partial \ket{\psi}}{\partial x_j}$,
\begin{equation}
    \frac{\partial \ket{\psi}}{\partial x_j}=\sum_{k=1}^d U^{(d)}...iw_j^{(k)}P_jU^{(k)}...U^{(1)}\ket{0}.
\end{equation}
As the norm of operators $U^{(k)}$ and $P_j$ is equal to 1, we can bound $\left|\frac{\partial \ket{\psi}}{\partial x_j}\right|_2$ by
\begin{align}
\left|\frac{\partial \ket{\psi}}{\partial x_j}\right|_2 & \leq \sum_{k=1}^d \left|U^{(d)}...w_j^{(k)}P_j U^{(k)}...U^{(1)}\ket{0}\right|_2\cr
& = \sum_{k=1}^d\left|w_j^{(k)}\right|\cr
& \leq d\max_j\left|\bm{w}_j\right|.
\end{align}
Hence, we can bound $\left|\frac{\partial l}{\partial x_j}\right|$ by 
\begin{align}
    \left|\frac{\partial l}{\partial x_j}\right| &\leq 2d\max_j|\bm{w}_j|\cdot\left|M\right|\cdot(\left|M\right||+\max|y|). 
\end{align}
Provided the feature space and the label space are bounded, the loss function is $\lambda$-Lipschitz continuous with
\begin{equation}
\lambda=2d\sqrt{d_{\bm{x}}}\max_j\left|\bm{w}_j\right|\cdot \left|M\right|(\left|M\right|+\max|y|).
\end{equation}

\end{proof}


\subsection{Proof of Theorem \ref{thm:qnn-coreset}}

We neglect the residual training error in coreset from now on. For the $n_c$-class classification problem, we assume that there are class-specific regression functions that represent the probability of a class label in the neighborhood of a feature vector $\eta_{c}({\bm{x}})=p(y=c |{{\bm{x}}})$ and that they satisfy the property of being $\lambda_\eta$-Lipschitz continuous with respect to the feature space of $\bm{x}$. Then we present Lipschitz constants for the loss functions of QNN and derive the upper bound for the generalization error.

\begin{proof}
By triangular inequality, 
\begin{equation}
    G_c^{\text{QNN}}=\left|R - R_c^{\text{QNN}}\right| \leq \left|R - R_e^{\text{QNN}}\right|+\left|R_e^{\text{QNN}} - R_c^{\text{QNN}}\right|.
\end{equation}
Since $\left|R-R_e^{\text{QNN}}\right|$ is bounded by Theorem \ref{thm:qnn-generalization}, 
\begin{equation}
|R-R_{e}^{\text{QNN}}| < \mathcal{O}\left(\sqrt{\frac{m \log(m)}{N_t}}+\sqrt{\frac{\log(1/\delta)}{N_t}}\right).
\end{equation}
Let us attempt to relate $R_e$ with coreset. Since we assume the training loss over coreset $R_c^{\text{QNN}}$ is equal to zero, only $R_e^{\text{QNN}}$ should be considered. Here we allow differently labelled samples to appear mixed in a local region of the feature space. Instead of regarding the probability of each label jumping between 0 and 1, we use the local class-specific regression function $\eta$ to describe the probability of each label close to a point.
\begin{align}
R_e^{\text{QNN}} &= \frac{1}{|{\cal{X}}|} \sum_{{{\bm{x}}}_j \in {\cal{X}}} \sum_{c_i} ( \eta_{c_i}  ({\bm{x}_j}) \cdot l(f_{\bm{\theta}}({{\bm{x}}}_j),c_i))\cr
&= \mathbb{E}_{{{\bm{x}}}_j \in {\cal{X}}} \sum_{c_i} ( \eta_{c_i}({\bm{x}_j}) \cdot l(f_{\bm{\theta}}({{\bm{x}}}_j),c_i)).
\end{align}
For each $({\bm{x}_j},y_j)$ drawn from the training set, we are guaranteed with the nearest center $({\bm{x}_c},y_j)$ of the same class in the coreset with distance less than $\delta_c$. If locally the class-specific regression function $\eta_{c_j} \neq 0$ for another label $c_i$ with $c_i \neq y_j$, then we would have at least one differently labeled sample within $\delta_c$ from ${\bm{x}_c}$ and $2\delta_c$ from ${\bm{x}_j}$. This other-class sample in the same circle, in turn, guarantees an other-class center ${\bm{x}_c}(c_i)$ within $3\delta_c$ from ${\bm{x}_j}$. For each $\bm{x}_j$, we have
\begin{align}
&\sum_{c_i} \eta_{c_i}({\bm{x}_j}) \cdot l(f_{\bm{\theta}}({\bm{x}_j}),c_i)\cr
= &\sum_{c_i} \eta_{c_i}({\bm{x}_j})\cdot \left[ l(f_{\bm{\theta}}(\bm{x}_j),c_i)-l(f_{\bm{\theta}}(\bm{x}_j(c_i)),c_i)\right] + \sum_{c_i} \eta_{c_i}(\bm{x}_j) \cdot l(f_{\bm{\theta}}(\bm{x}_c(c_i)),c_i)\cr
=&\sum_{c_i} \eta_{c_i}({\bm{x}_j})\cdot \left[ l(f_{\bm{\theta}}(\bm{x}_j),c_i)-l(f_{\bm{\theta}}(\bm{x}_j(c_i)),c_i)\right] + \sum_{c_i} \left[\eta_{c_i}(\bm{x}_j)-\eta_{c_i}(\bm{x}_c(c_i))\right] \cdot l(f_{\bm{\theta}}(\bm{x}_c(c_i)),c_i)\cr 
& + \sum_{c_i} \eta_{c_i}(\bm{x}_c(c_i))\cdot l(f_{\bm{\theta}}(\bm{x}_c(c_i)),c_i).
\end{align}
The last term is the loss of the center weighted by the local probability of being same-class as that center. For a specific center, it will be called by both same-class samples and other-class samples that find it closest in the coreset, the number of calls is approximately proportional to the number of samples covered by the center. With the factor $\eta_{c_i}(\bm{x}_c(c_i))$ in each call, the corresponding weight for a center in the coreset should be the number of samples covered $\times$ local probability of being center-class, i.e., the number of same-class samples covered by the center. This justifies our choice of the weight in Eq. (\ref{eq:gamma}). Assuming the weighted loss is trained to zero, the form of the previous expression is then simplified to
\begin{align}
\sum_{c_i} \eta_{c_i}(\bm{x}_j) \cdot l(f_{\bm{\theta}}(\bm{x}_j),c_i)&=\sum_{c_i} \eta_{c_i}(\bm{x}_j) \cdot \left[ l(f_{\bm{\theta}}(\bm{x}_j),c_i)-l(f_{\bm{\theta}}(\bm{x}_c(c_i)),c_i)\right] \cr
&+\sum_{c_i} \left[\eta_{c_i}(\bm{x}_j) - \eta_{c_i}(\bm{x}_c(c_i)\right] \cdot l(f_{\bm{\theta}}(\bm{x}_c(c_i)),c_i).
\end{align}
The two parts in the last equation can be separately bounded by the Lipschitz continuity $\lambda_{\zeta}$ and $\lambda_l$ of $\eta_{c_i}$ and $l$, and assume the $l$ is bounded by $L$. Substituting into the corresponding constants, we have
\begin{align}
\mathbb{E}_{{{\bm{x}}}_j \in {\cal{X}}} \sum_{c_i} ( \eta_{c_i}({\bm{x}_j}) \cdot l(f_{\bm{\theta}}({{\bm{x}}}_j),c_i)) \leq & (\sum_{c_i\neq y_j} \eta_{c_i}(\bm{x}_j)\cdot3\delta_c + \eta_{y_j}(\bm{x}_j) \cdot \delta_c) \lambda_l + ((n_c-1)\cdot3\delta_c + \delta_c)L\lambda_{\eta}   
\cr
\leq& (\sum_{c_i} \eta_{c_i}(\bm{x}_j))3\delta_c\lambda_l+ (3n_c-2)\delta_c\lambda_{\eta} L  \cr
=& 3\delta_c\lambda_l+ (3n_c-2)\delta_c\lambda_{\eta} L.  \cr
\end{align}
According to Hoeffding's bound and assumption of zero coreset training loss, we have the coreset error of QNN with probability $1-\delta$,
\begin{equation}
    \left|R_e^{\text{QNN}} - R_c^{\text{QNN}}\right| \leq 3\delta_c\lambda_l + (3n_c-2)\delta_c\lambda_\eta L + \sqrt{\frac{L^2\log(1/\delta)}{2n}}.
\end{equation}
Combing the bound of the $\left|R - R_e^{\text{QNN}}\right|$, we have the generalization bound of QNN on coreset with probability $1-\delta$
\begin{equation}
    G_c^{\text{QNN}}\leq \mathcal{O}\left(\sqrt{\frac{m \log(m)}{N_t}}+\sqrt{\frac{\log(1/\delta)}{N_t}}+\delta_c (\lambda_\eta L n_c + d\sqrt{d_{\bm{x}}}\max_j|\bm{w}_j|\left|M\right|(\left|M\right|+\max|y|)\right).
\end{equation}
\end{proof}

\begin{theorem}\label{thm:qkernel-generalization-learning}{\bf{(Generalization bound for learning with quantum kernel \cite{abbas2021power})}} Suppose the parameters $\bm{w}$ of the hypothesis $f_{\bm{w}}$ is optimized to be $\bm{w^*}$ by minimizing the loss $l(f_{\bm{w}}(\bm{x})),y)=|\min(1,\max(-1,\sum_j w_j\kappa(\bm{x}_j,\bm{x}))) - y|$ on a training set $\{{{\bm{x}}}_i, y_i\}_{i=1}^{N_t}$. Then, according to Theorem 2 in \cite{Huang2021PowerOD}, the generalization error of the optimized hypothesis is bounded by the following form with probability at least $(1-\delta)$ over the choice of $i.i.d$ training data ${\cal{S}}$ from $\mathcal{Z}$:
\begin{equation}\label{eq:qkernel_learning_bound}
|R-R_{e}^{\text{qkernel}}| < \mathcal{O}\left(\sqrt{\frac{\ceil{\|\bm{w}\|^2}}{N_t}}+\sqrt{\frac{\log(4/\delta)}{N_t}}\right).
\end{equation}
\end{theorem}

\begin{lemma}\label{lemma:kernel}
The function $l(f_{\bm{w}}(\bm{x})),y)=|\min(1,\max(-1,\sum_j w_j\kappa(\bm{x}_j,\bm{x}))) - y|$, where $\kappa(\bm{x},\bm{x}')$ is the quantum kernel function, is $\lambda$-Lipschitz continuous in feature space with 
\begin{equation}
\lambda=2k_c\sqrt{d_{\bm{x}}}\max_i|w_i|\cdot(1+(N_q-1)r).
\end{equation}
\end{lemma}

\begin{proof}
This follows immediately from evaluating the norm of the weight vector in the mapped Hilbert space. The function $l$ can be viewed as $l=g \circ f$, where 
\begin{equation}
g(f)=|\min(1,\max(-1,f)) - y|,
\end{equation}
and 
\begin{equation}
f({\bm{w}},\bm{x}) = \sum_{j=1}^{k_c} w_j \kappa({\bm{x}_j},\bm{x}),
\end{equation}
where $\kappa:\mathbb{R}\times\mathbb{R}\rightarrow \mathbb{R}$ is the quantum kernel function and $k_c$ being the size of the coreset. The function $g$ has a Lipschitz constant of $1$ with respect to $f$. Concretely, it is defined as follows
\begin{equation}
\kappa(\bm{x}_j,\bm{x}_k)=\langle \phi(\bm{x}_j),\phi(\bm{x}_k) \rangle = |\langle0|U_{\bm{x}_j}^\dagger U_{\bm{x}_k}|0\rangle|^2,
\end{equation}
where $U_{\bm{x}}=(U(\bm{x})H^{\otimes N_q})^2 |0\rangle^{\otimes N_q}$, $U(\bm{x})=\exp(\sum_{j=1}^{N_q}\bm{x}_j\sigma^Z_j+\sum_{j,j'=1}^{N_q}\bm{x}_j\bm{x}_{j'}\sigma^Z_j\sigma^Z_{j'})$, and $N_q$ is system size of mapped state\\
For any fixed ${\bm{w}}$, we have
\begin{align}\label{eq:tmp0}
\left|f({\bm{w}},\bm{x})-f(\bm{w},\bm{x}')\right| & \leq \max_{\bm{x}}\sqrt{\sum_{j=1}^{d_{\bm{x}}}\left|\frac{\partial{f}}{\partial{x_j}}\right|^2}\cdot\left|\bm{x}-\bm{x}'\right| \cr
& \leq \sqrt{d_{\bm{x}}}\max_{j}\left|\frac{\partial{f}}{\partial{x_j}}\right|\cdot\left|\bm{x}-\bm{x}' \right|.
\end{align}

 Assume we train the SVM with quantum kernel on a coreset $\{ \bm{x}_i \}_{i=1}^{k_c}$, and let $\rho(\bm{x})=U_{\bm{x}}|0\rangle \langle 0|U^{\dagger}_{\bm{x}}$, we have
\begin{align}
    \left|\frac{\partial{f}}{\partial{x_j}}\right| &= 2\left|\sum_{k=1}^{k_c} w_k \mathcal{R}(\langle 0| \frac{\partial{U^{\dagger}(\bm{x})}}{\partial{x_j}}\rho(\bm{x}_k)U(\bm{x})|0\rangle)\right|\cr
    & \leq 2k_c\cdot \max_k |\bm{w}_k| \cdot \left|\langle 0| \frac{\partial{U^{\dagger}(\bm{x})}}{\partial{x_j}}\rho(\bm{x}_k)U(\bm{x})|0\rangle\right|\cr
    & = 2k_c\cdot \max_k |\bm{w}_k| \cdot \left|\frac{\partial{U(\bm{x})}}{\partial{x_j}}\right|_2. \label{eq:tmp1}
\end{align}
Since $U(\bm{x})=\exp(\sum_{j=1}^{N_q}\bm{x}_j\sigma^Z_j+\sum_{j,j'=1}^{N_q}\bm{x}_j\bm{x}_{j'}\sigma^Z_j\sigma^Z_{j'})$, we denote $U_s=\prod_{k=1}^{N_q}\exp(ix_k\sigma^Z_k)$ and $U_e=\prod_{k<l}^{N_q} U_{k,l}$, where $U_{k,l}=\exp(ix_k x_l\sigma^Z_k\sigma^Z_l)$,
then the norm of the derivative of $U(\bm{x})$ with respect to $x_j$ can be bounded,
\begin{align}
    \left|\frac{\partial{U(\bm{x})}}{\partial{x_j}}\right|_2 &= \left| \frac{\partial{U_s}}{\partial{x_j}} U_e + U_s \frac{\partial{U_e}}{\partial{x_j}} \right|\\
    &\leq \left|\frac{\partial{U_s}}{\partial{x_j}}\right|_2\cdot \left|U_e\right|_2+\left|U_s\right|_2\cdot \left|\frac{\partial{U_e}}{\partial{x_j}}\right|_2\\
    & \leq \left|i\sigma^Z_j\right| \cdot \left|\exp(ix_j\sigma^Z_j)\right| \cdot \left|\prod^{N_q}_{k\neq j}\exp(i x_k \sigma^Z_k)\right| \cr
    &+ \left|\sum_{k\neq j}^{{N_q}}i x_k\sigma_{k}^Z\sigma_{j}^ZU_{k,j}\cdot \prod_{m<n, (m,n) \neq (k,j) \text{or} (j,k)}^{N_q}U_{m,n}\right|\\
    & \leq 1+(N_q-1)\max_j|x_j|. \label{eq:tmp2}
\end{align}
Substitute Eq.(\ref{eq:tmp2}) to Eq.(\ref{eq:tmp1}), we have
\begin{equation}\label{eq:tmp3}
    \left|\frac{\partial{f}}{\partial{x_j}}\right| \leq 2k_c\max_i|\bm{w}_i|\cdot(1+(N_q-1)r).
\end{equation}
Assuming $r$, the maximum feature value of $\bm{x}$, is bounded and substituting Eq.(\ref{eq:tmp3}) to Eq.(\ref{eq:tmp0}), we can conclude the function $l(f_{\bm{w}}(\bm{x})),y)$ is $2k_c\sqrt{d_{\bm{x}}}\max_i|\bm{w}_i|\cdot(1+(N_q-1)r)$-Lipschitz continuous.
\end{proof}

\subsection{Proof of Theorem \ref{thm:qkernel-coreset}}

\begin{proof}
The proof of Theorem \ref{thm:qkernel-coreset} is similar to the proof of Theorem \ref{thm:qnn-coreset}, replacing the relevant Lipschitz constant and generalization bound from Lemma \ref{lemma:kernel} and Theorem \ref{thm:qnn-coreset} respectively. 
\end{proof}

\end{document}